\documentclass[11pt,a4paper]{article}
\usepackage[dvips]{graphicx}
\usepackage[utf8]{inputenc}
\usepackage{enumerate}
\usepackage[spanish,english]{babel}
\usepackage{amsthm,amssymb,amsfonts,amsmath}
\usepackage[ruled,linesnumbered]{algorithm2e}
\usepackage{lineno}

\usepackage{xcolor}
\usepackage{cite}
\usepackage{pifont}
\usepackage{enumerate}

\newtheorem{theorem}{Theorem}

\newtheorem{proposition}{Proposition}

\newtheorem{lemma}{Lemma}
\newtheorem{corollary}{Corollary}
\newtheorem{observation}{Observation}

\graphicspath{{figures/}}

\newcommand{\N}{\ensuremath{\mathcal{N}_{\ell}}\xspace}
\newcommand{\PP}{\ensuremath{\mathcal{P}_{\ell}}\xspace}
\newcommand{\diam}{\ensuremath{\rm diam}\xspace}
\DeclareMathOperator*{\argmax}{arg\,max}

\newcommand\myparagraph[1]{\vspace{6pt}\noindent \textbf{#1}.}

\title{Computing optimal shortcuts for networks}

\author{\normalfont\fontsize{12}{14}\selectfont
Delia Garijo$^1$\and
Alberto M\'arquez$^1$\and
Natalia Rodr{\'i}guez$^2$\and
Rodrigo I. Silveira$^3$
}

\date{}
\begin{document}

\maketitle

\addtocounter{footnote}{1}  \footnotetext{Departamento de Matem\'atica Aplicada I, Universidad de Sevilla, Avda. Reina Mercedes s/n, 41012, Sevilla, Spain. Emails: \{dgarijo, almar\}@us.es}
\addtocounter{footnote}{1} \footnotetext{Departamento de Computaci\'on, Universidad de Buenos Aires, Pabellón I, Ciudad Universitaria, C1428EGA, Buenos Aires, Argentina. Email: nrodriguez@dc.uba.ar}
\addtocounter{footnote}{1} \footnotetext{\textbf{Corresponding author.} Departament de Matem\`{a}tiques, Universitat Polit\`{e}cnica de Catalunya, Jordi Girona 1-3, 08034, Barcelona, 
Spain. Email: rodrigo.silveira@upc.edu}

\begin{abstract}
We study augmenting a plane Euclidean network with a segment, called a \emph{shortcut}, to minimize the largest distance between any two points along the edges of the resulting network. Problems of this type have received considerable attention recently, mostly for discrete variants of the problem.
We consider a fully continuous setting, where the problem of computing distances and placing a shortcut is much harder as all points on the network, instead of only the vertices, must be taken into account. We present the first results on the computation of optimal shortcuts for general networks in this model: a polynomial time algorithm and a discretization of the problem that leads to an approximation algorithm. We also improve the general method for networks that are paths, restricted to two types of shortcuts: those with a fixed orientation and simple shortcuts.

\end{abstract}




\section{Introduction}

A fundamental task in network analysis, especially in the context of geographic data (for instance, for networks that model roads, rivers, or train tracks), is analyzing how an existing network can be improved.
This can arise in many different contexts: in relation to facility location analysis, for example, to guarantee a certain maximum travel time from any point on the network to the nearest hospital, or in road network design problems, to decide where to add road segments to reduce network congestion~\cite{YangBell98}.

In fact, due to the many and varied applications of this type of network optimization problem, there is a vast amount of work in areas like operational research and  transportation science, which deals with different aspects of the so-called Network Design Problem: deciding how to add or modify network connections (e.g., roadway segments) in order to improve aspects such as travel time, capacity or connectivity (see, for instance \cite{YangBell98,FEREMANS20031,FARAHANI2013281}).

One of the simplest approaches to model networks like the ones above is as a \emph{geometric network}: an undirected graph whose vertices are points in $\mathbb{R}^2$ and whose edges are straight-line segments connecting pairs of points. Moreover, in many applications, it is reasonable to assign lengths to the edges equal to the Euclidean distance between their endpoints.
These are called \emph{Euclidean networks}.
When, in addition, there are no crossings between edges, the Euclidean network is said to be \emph{plane}.
Many problems in geographic analysis, for instance, those involving certain transportation networks, can be reasonably  modeled with a plane Euclidean network.
In the following, we shall simply write \emph{network}, it being understood as plane and Euclidean.

One of the most fundamental ways to improve a network is by adding edges.
This increases the connectivity of the network and potentially can decrease travel times and congestion.
The most studied criteria to measure network improvement, in the geometric setting, are related to distances.
Particularly important is the maximum distance, or \emph{diameter} of the network, which provides an upper bound on the distance between any two network points.
Another important distance-related criterion in this context is the \emph{dilation}, which captures the maximum detour between two points on the network.

In this paper, we focus on the problem of adding segments to a network in order to improve its diameter.
This can be seen as a variant of the Diameter-Optimal-$k$-Augmentation problem for edge-weighted geometric graphs, which consists in inserting $k$ additional edges into an edge-weighted geometric graph, while minimizing the largest distance in the resulting graph.
While in most augmentation problems the additional edges are inserted between two existing vertices of the graph, we consider a continuous version,  for $k=1$: the endpoints of our inserted segment, called \emph{shortcut}, can be any two points (not necessarily vertices) on the network.  In addition, our goal is to find an \emph{optimal} shortcut: one minimizing the diameter of the resulting network, over all possible shortcuts. 
We point out here that a segment will be considered a shortcut only if its insertion improves the diameter of the resulting network. 
The fact that all points on the network must be taken into account turn the computation of distances and the placement of a shortcut into a much harder problem.
For instance, since the resulting network also includes the points on the shortcut inserted, we can run into the somewhat counterintuitive situation that the insertion of a segment  worsens the diameter of the resulting network.

In the continuous setting, two major variants of the problem arise, depending on how the shortcut is inserted into the network.
In the first variant, which we call \emph{highway model}, the crossings between the shortcut and the network edges do not form new network vertices: a path can only enter and leave the shortcut through its endpoints.
In contrast, in the \emph{planar model}, every crossing creates a new vertex, which can be used by paths in the network.
Figure~\ref{fig:highway} illustrates the difference between the two models.

\begin{figure}[ht]
\begin{center}
\includegraphics{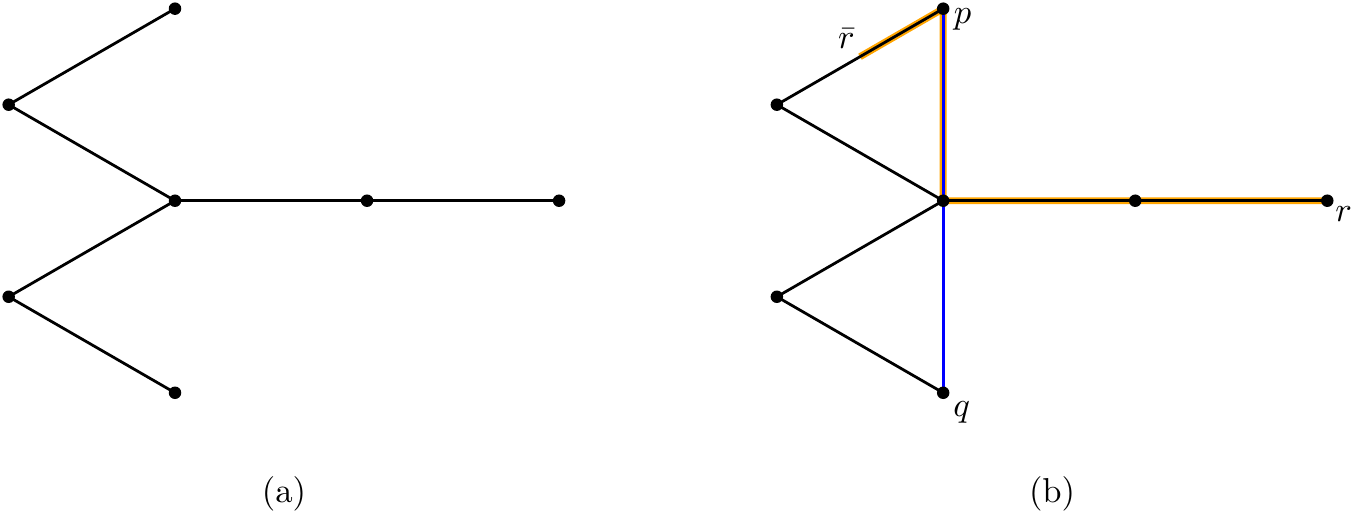}
\end{center}
\caption{Example of network with diameter 4 (each edge has unit length): (a) in the highway model, no single segment insertion can improve the diameter, (b) in the planar model, segment $pq$ is a shortcut, since its insertion reduces the diameter to $d(r,\bar{r})=3.5$.}\label{fig:highway}
\end{figure}

We deal in this paper with optimal shortcuts in the planar model. This model is more general, and is applicable to a wider range of situations, like the addition of segments to road or pedestrian networks.
From a theoretical point of view, the difference between the highway and planar model is important.
The latter results in more complex problems, since the fact that a shortcut can be used only in part, implies that the structural information on how the distances in the network change after adding a segment is more difficult to maintain.
Moreover, as we show in this work, many intuitive properties of shortcuts do not hold in the planar model anymore.

\myparagraph{Related work}

The motivation of this work lies on the well-studied Road Network Design Problem, which has been subject of much attention in many areas, including operational research, transport geography, and engineering, and has been recognized as one of the most difficult and challenging problems in transport~\cite{YangBell98}.
We refer to~\cite{YangBell98,FEREMANS20031,FARAHANI2013281} for relevant surveys on this general problem, for their urban and road variants.
The network design problem is extremely complex, involving plenty of different aspects such as travel demand, road capacity, traffic flow, and congestion.
In this work we focus on a much simplified geometric scenario, with the goal of exploring some of the algorithmic aspects of this complex problem. 
For that reason, our problem is mostly related graph augmentation problems, and in particular to those where the graph is geometric.

There has been a considerable amount of work devoted to the non-geometric version of the Optimal-$k$-Augmentation problem, see for instance \cite{FGGM14,GGKSS2015}. However, there are very few results on augmentation problems over plane geometric graphs. Farshi et al.~\cite{FGG05} considered plane Euclidean networks in $\mathbb{R}^d$ and designed approximation algorithms to minimize the dilation when inserting one edge. The analogous problem for networks embedded in a metric space is studied in \cite{LW08,W10}. See the survey~\cite{ht-13} for more on augmentation problems over plane geometric graphs.

In the geometric and continuous setting, most known results on the problem studied in this paper consider the highway model and certain classes of graphs. For paths, De Carufel et al.~\cite{DBLP:conf/swat/CarufelMS16} gave an algorithm to find an optimal shortcut in linear time, and also optimal pairs of shortcuts (i.e., $k=2$) for convex cycles. Trees have been studied in a recent follow-up work~\cite{CGSS-17}, which presents an algorithm to find an optimal shortcut for a tree of size $n$ in $O(n\log n)$ time. For circles, very recently Bae et al.~\cite{BaeBCGL17} have analyzed how to add up to seven shortcuts in an optimal way.

For the planar model much lees is known. Yang~\cite{yang} designed three different approximation algorithms to compute an \emph{optimal} shortcut for certain types of paths. C\'aceres et al.~\cite{CGGMPR17} were the first to deal with general networks, for which they show that one can find a shortcut in polynomial time if one exists, but they do not compute an optimal one. Note that there are networks whose diameter cannot be improved by adding only one segment (e.g., a cycle).

\myparagraph{Our results}
We present the first study of optimal shortcuts in the planar model for general networks, and several improved results for paths.
An important contribution of our work is to highlight many important differences between the highway and planar models, the latter resulting in considerably harder problems.
In Section \ref{general_networks}, we give a polynomial time algorithm to compute an optimal shortcut, if one exists, for a general network.
Moreover, we present a discretization of the problem that immediately leads to an approximation algorithm, generalizing an existing result for paths~\cite{yang}.
Section \ref{sec:path} focuses on paths: we first show that the diameter of a path network after adding a shortcut can be computed in $\Theta(n)$ time.
Then we improve the general method of Section \ref{general_networks} for shortcuts of any fixed direction.
Finally, we study simple shortcuts, a variant that has been considered before \cite{CGGMPR17,yang}, which has applications in settings where the added edge cannot intersect the existing network. We show that an optimal simple shortcut can be found, if one exists, in $O(n^2)$ time.

\subsection{Preliminaries}

We will use $\mathcal{N}=(V(\mathcal{N}), E(\mathcal{N}))$ to denote a network with $n$ vertices, and  $\N$ for its \emph{locus}, the set of all points of the Euclidean plane that are on $\mathcal{N}$. Thus,
$\N$ is treated indistinctly as a network or as a closed point set. When $\N$ is a path, we use $\PP$ instead of $\N$. Further, we write $a\in \N$ for a point $a$ on $\N$, and $V(\mathcal{N})\subset \N$.

A \emph{path} $P$ between two points $a,b$ on $\N$ is a sequence $au_1\dots u_kb$ where $u_1u_2,\dots,u_{k-1}u_k\in E(\mathcal{N})$, $a$ is a point on an edge ($\neq u_1u_2$) incident to $u_1$, and $b$ is a point on an edge ($\neq u_{k-1}u_k$) incident to $u_k$. We use $|P|$ to denote the \emph{length} of $P$, i.e., the sum of the lengths of all edges $u_iu_{i+1}$ plus the lengths of the segments $au_1$ and $bu_k$.
The length of a shortest path from $a$ to $b$ is the \emph{distance} between $a$ and $b$ on ${N}_{\ell}$. This distance is written as $d_{{N}_{\ell}}(a,b)$ or $d(a,b)$ when the network is clear, and whenever $ab\notin E({N}_{\ell})$, it is larger than  $|ab|$, the Euclidean distance between the points.

The \emph{eccentricity} of a point $a\in \N$ is ecc$(a)=\max_{b \in \N} d(a,b)$, and the {\em diameter} of $\N$, also called \emph{generalized} or \emph{continuous diameter} of $\mathcal{N}$ \cite{CG82,DBLP:conf/swat/CarufelMS16,CGSS-17}, is $ {\rm diam}(\N) = \max_{a\in \N} {\rm ecc}(a).$ Two points $a,b \in \N$ are {\em diametral} whenever $d(a,b)={\rm diam}(\N)$, and a shortest path connecting $a$ and $b$ is then called \emph{diametral path}.

The diameter of \N, ${\rm diam}(\N)$, can be computed in quadratic time \cite{CG82,CGGMPR17}.
Furthermore, the diametral pairs of $\N$ are either (i) two vertices, (ii) two points on distinct non-pendant edges\footnote{An edge $uv\in E(\mathcal{N})$ is {\em pendant} if either $u$ or $v$ is a {\em pendant} vertex (i.e., has degree 1).}, or (iii) a pendant vertex and a point on a non-pendant edge~\cite[Lemma 6]{CGGMPR17}. Thus, with some abuse of notation, in Section \ref{general_networks}, we will say that a {\em diametral pair} $\alpha,\beta\in V(\mathcal{N})\cup E(\mathcal{N})$ may be (i) vertex-vertex, (ii) edge-edge, or (iii) vertex-edge.

A {\em shortcut} for $\N$ is a segment $s$ with endpoints on $\N$ such that its insertion improves the diameter of the resulting network, that is, $ {\rm diam}(\N \cup s)<{\rm diam}(\N).$
We say that shortcut $s$ is {\em simple} if its two endpoints are the only intersection points with $\N$, and $s$ is \emph{maximal} if it is the intersection of a line and $(\N \cup s)$, i.e., $s=(\N \cup s) \cap \ell$, for some line $\ell$; see Figure \ref{fig:highway}(b) for an example of a non-simple and maximal shortcut. A shortcut is \emph{optimal} if it minimizes ${\rm diam}(\N \cup s) $ among all shortcuts $s$ for $\N$.

\section{General networks}\label{general_networks}

The main result in \cite{CGGMPR17} states that one can always determine in polynomial time whether a network $\N$ has a shortcut (and compute one, in case of existence).
In this section, we first prove the analogous result for optimal shortcuts. Our proof uses some ideas in \cite{CGGMPR17} but captures the property of being optimal with a much shorter argument based on some functions defined in Lemma \ref{lem:ecc_func} below.

Let $\alpha,\beta\in V(\mathcal{N})\cup E(\mathcal{N})$, and let $e=uv$ and $e'=u'v'$ be two edges of $\mathcal{N}$. When $\alpha$ is an edge, we use ${\rm ecc}(u,\alpha)$ to indicate the maximum distance from $u$ to the points on $\alpha$ (analogous for $\beta$ and the remaining endpoints of $e$ and $e'$); if $\alpha$ is a vertex, ${\rm ecc}(u,\alpha)=d(u,\alpha)$. In general, ${\rm ecc}(\alpha,\beta)={\rm max}_{t\in \alpha, z\in \beta} d(t,z)$.

\begin{lemma}
\label{lem:ecc_func}
Let $ y=ax+b$ be a line intersecting edges $e=uv$ and $e'=u'v'$ on points $p$ and $q$, respectively, and let $\alpha,\beta\in V(\mathcal{N})\cup E(\mathcal{N})$. For each pair $(w,z)$ with $w\in \{u,v\}$ and $z\in \{u',v'\}$, function $f_{\alpha,\beta}^{w,z}(a,b)={\rm ecc}(w,\alpha)+|wp|+|pq|+|qz|+{\rm ecc}(z,\beta)$ is linear in $b$.
\end{lemma}

\begin{proof}
 The only term in the expression of $f_{\alpha,\beta}^{w,z}(a,b)$ that is not linear is $|pq|$ (the values ${\rm ecc}(w,\alpha)$ and ${\rm ecc}(z,\beta)$ are fixed, and $|wp|$ and $|qz|$ are linear). Thus, it suffices to prove that there is a linear change when line $y=ax+b$ is moved parallel to itself (changing slightly $b$), which is a straightforward consequence of the Side Splitter Theorem.
\end{proof}

The following theorem is the optimality version of Theorem 8 in \cite{CGGMPR17}.

\begin{theorem}\label{th:polynomial}
It is possible to determine in polynomial time whether a network $\mathcal{N}_{\ell}$ admits an optimal shortcut, and compute one in case of existence.
\end{theorem}

\begin{proof}
As explained in~\cite[Proposition 7]{CGGMPR17}, the search space can be split into a polynomial number of regions; we include a description here for the sake of completeness.

Two lines are \emph{equivalent} if the half-planes to the right (resp., to the left) that they define contain the same vertices of $\N$.
There are $O(n^2)$ classes of equivalent lines as one can associate four lines to every pair of vertices $u,v$, say,
lines $m_{u^+v^+}$ and $m_{u^-v^-}$ parallel to segment $uv$ and leaving $u$ and $v$ in the same half-plane, and lines $m_{u^+v^-}$, $m_{u^-v^+}$ leaving $u,v$ in different half-planes. These lines must be placed sufficiently close to $u$ and $v$ as Figure \ref{fig:th}(a) shows. Thus, every class of equivalent lines has a representative in the set $\{m_{u^+v^+}, m_{u^-v^-}, m_{u^+v^-}$, $m_{u^-v^+} \, | \, u,v\in V(\mathcal{N})\}$, whose cardinality is $O(n^2)$.

Given a line $ m$ that crosses two fixed edges $e,e' \in E(\mathcal{N})$, let  ${\cal P}_{e,e'}(m)$ be the set of  lines  equivalent to $m$. This set can be seen as the set of segments determined by the corresponding intersections of the lines in ${\cal P}_{e,e'}(m)$ with edges $e$ and $e'$; thus, the region of the plane that ${\cal P}_{e,e'}(m)$ defines has the shape of an hourglass~\cite[Section 3.1]{CG89}, see Figure \ref{fig:th}(b). By the argument above, there are  $O(n^2)$ regions ${\cal P}_{e,e'}(m)$ per each pair of edges $e,e'\in E(\mathcal{N})$, which gives a total of $O(n^4)$ regions.

\begin{figure}[ht]
\begin{center}
\begin{tabular}{ccccccc}
\includegraphics[width=0.35\textwidth]{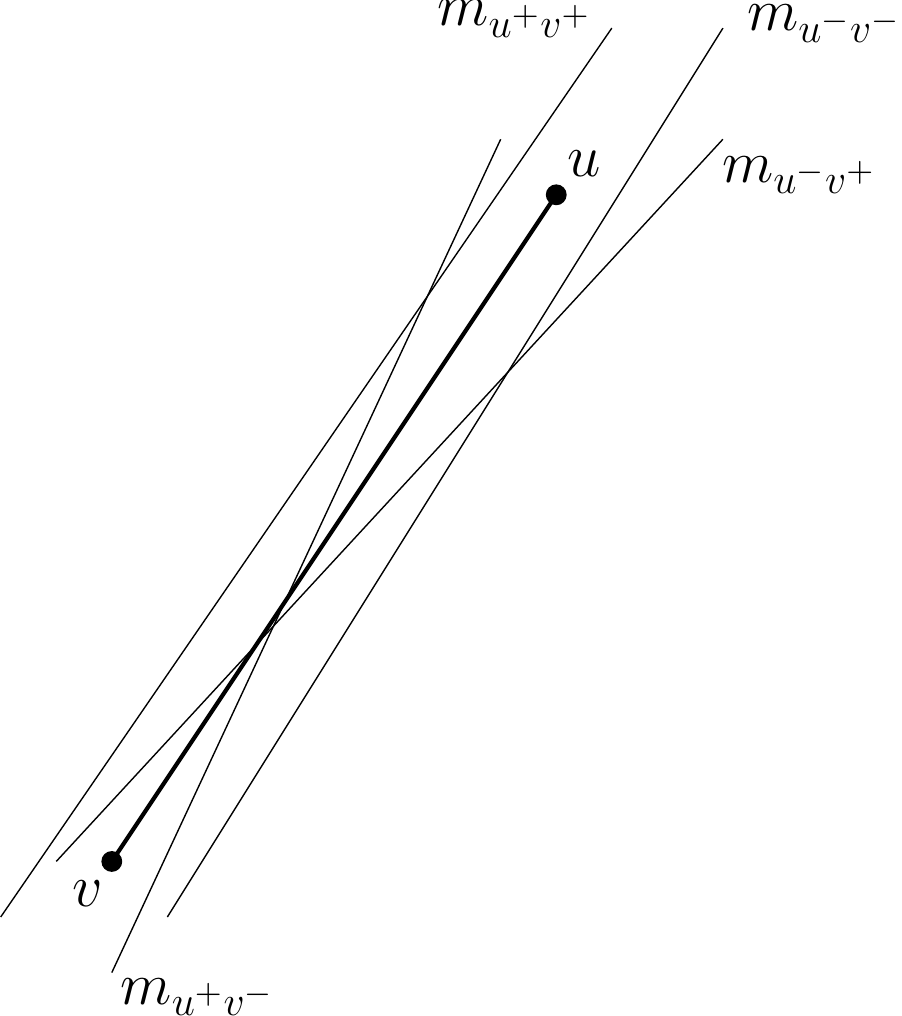}
&  &  &  & & &
\includegraphics[width=0.30\textwidth]{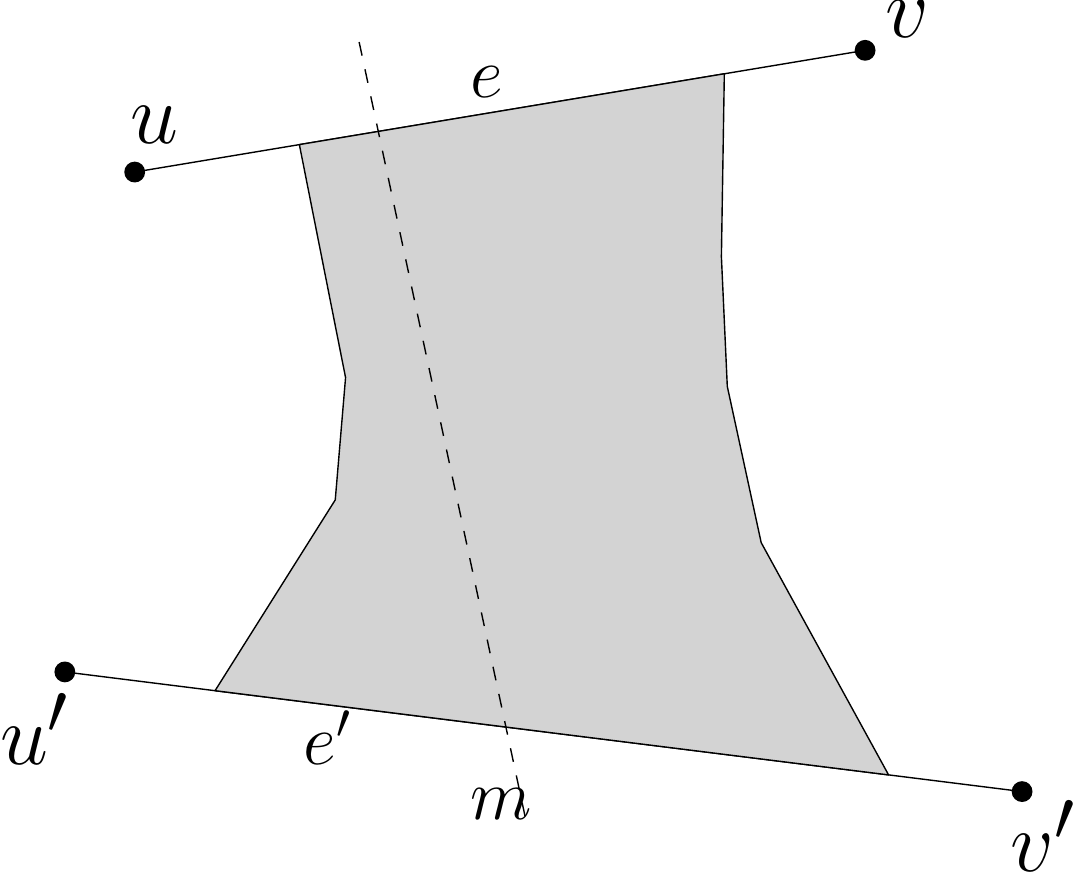}\\
(a)& & & & & & (b)
\end{tabular}
\end{center}
\caption{(a) Lines associated to vertices $u$ and $v$. (b) region ${\cal P}_{e,e'}(m)$.}\label{fig:th}
\end{figure}

 Consider a line $r \equiv y=ax+b$ in a region ${\cal P}_{e,e'}(m)$ and the set $I=\{e=e_0, e_1, \ldots, e_{k}, e_{k+1}=e'\}$ of edges that it intersects in between $e$ and $e'$; let $e_i=u_iv_i$ and $p_i=r\cap e_i$. For a fixed diametral pair $\alpha, \beta\in V(\mathcal{N})\cup E(\mathcal{N})$, a function of the type $f_{\alpha,\beta}^{w,z}(a,b)$ (see Lemma \ref{lem:ecc_func}) with $w\in \{u_i,v_i\}$ and $z\in \{u_j,v_j\}$, $0\leq i\neq j\leq k+1$, describes ${\rm ecc}(\alpha,\beta)$ when using a path that passes through vertices $w, z$ and contains segment $p_ip_j$.

Geometrically, functions $f_{\alpha,\beta}^{w,z}(a,b)$ can be viewed as lines in the $(b, f_{\alpha,\beta}^{w,z}(a,b))$-plane. By Lemma \ref{lem:ecc_func}, they are linear in $b$ and so, in order to compute them in an hourglass-shaped region ${\cal P}_{e,e'}(m)$, it suffices to obtain their evaluation only in those values of $(a,b)$ corresponding to the lines that form the boundary of the region, which is defined by $O(n)$ vertices.
 Note that for fixed $\alpha,\beta, w,z$, all functions $f_{\alpha,\beta}^{w,z}(a,b)$ have an important part in common, as illustrated in Figure~\ref{fig:running}. Thus, the common values of the distances from $\alpha$ to $z$ and from $w$ to $\beta$ only need to be computed once, which can be done in quadratic time. Each function will be different only in the portion of the path from $w$ to $z$ that depends on the segment defined by $(a,b)$, and this distance can be computed in constant time. This results in $O(n^6)$ running time for computing the $O(n^4)$ functions $f_{\alpha,\beta}^{w,z}(a,b)$ in a region ${\cal P}_{e,e'}(m)$.

\begin{figure}[t]
\centering
\includegraphics[width=0.5\textwidth]{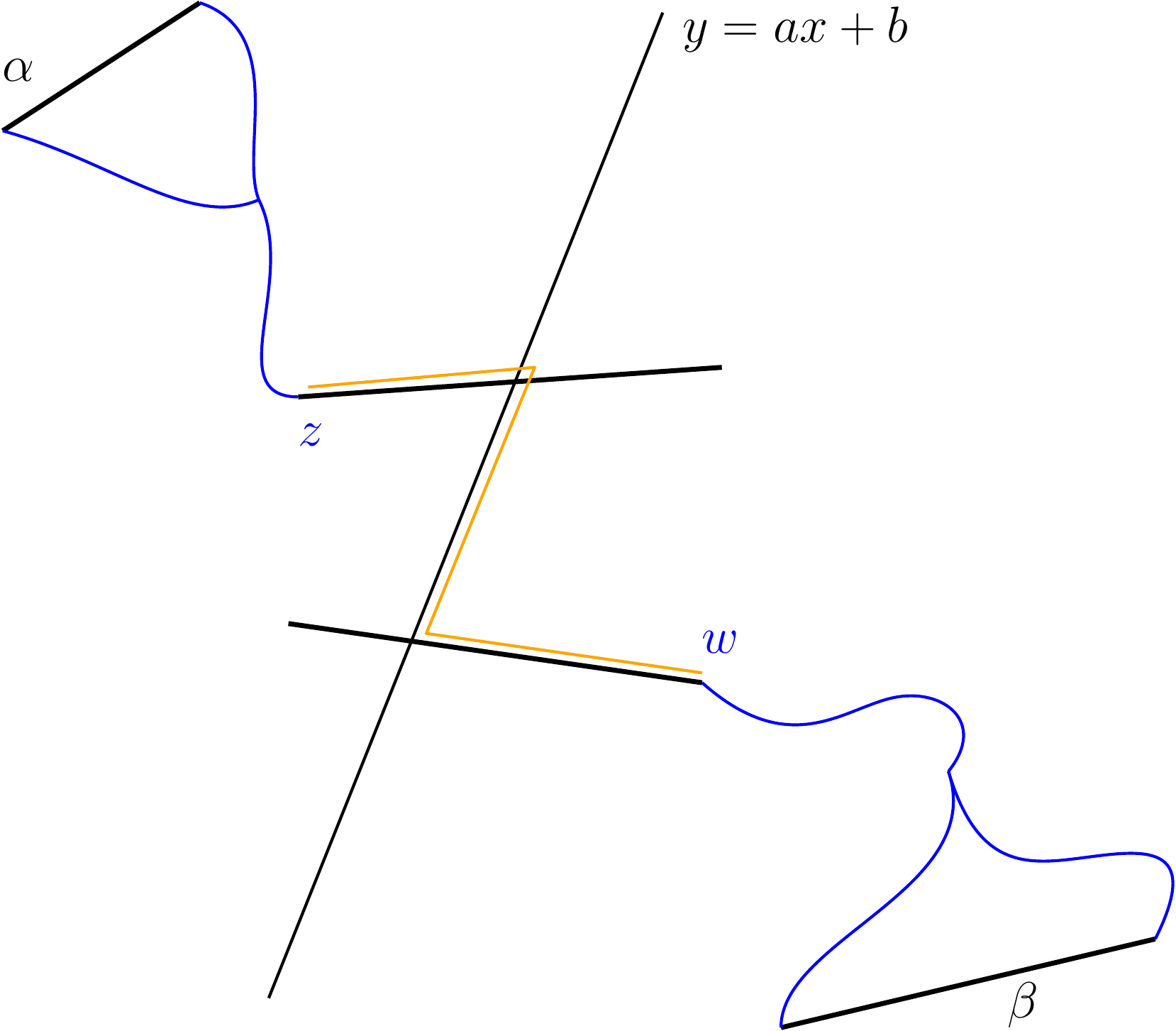}
\caption{The blue paths are common to all functions $f_{\alpha,\beta}^{w,z}(a,b)$ for fixed values of $\alpha,\beta, w,z$.}
\label{fig:running}
\end{figure}

An optimal shortcut for $\N$ in $\mathcal{P}_{e,e'}(m)$ is given by the minimum of the upper envelope of the set of lines $f_{\alpha,\beta}^{w,z}(a,b)$. Note that shortcuts must satisfy that ${\rm ecc}(t)<{\rm diam}(\N)$ for every $t\in s$, since, by definition, any shortcut improves the diameter of \N, and so all segments $p_ip_j$ must be included in the set of $O(n^2)$ diametral pairs $\alpha, \beta\in V(\mathcal{N})\cup E(\mathcal{N})$. Applying the same argument to the $O(n^4)$ regions $\mathcal{P}_{e,e'}(m)$, the result follows.
 \end{proof}

The importance of the preceding approach relies in the fact that the problem is polynomial.
However, its running time is very high as it involves $O(n^4)$ functions $f_{\alpha,\beta}^{w,z}(a,b)$ that must be computed, and this has to be done for each of the $O(n^4)$ regions $\mathcal{P}_{e,e'}(m)$.
Moreover, each evaluation of $f_{\alpha,\beta}^{w,z}(a,b)$ takes $O(n^2)$ time.
All in all, its running time amounts to $O(n^{10})$.

\subsection{Discretizing the set of possible shortcuts: approximation}

In light of the high running time of the previous approach, it becomes interesting to look for faster approximation algorithms.
Moreover, given the continuous nature of the problem, it is natural to wonder to what extent the problem can be discretized.
In other words, how good can shortcuts be if we restrict them to some discrete collection of segments?
The most natural choice for such a collection is probably the segments defined by pairs of vertices of $\N$, but this choice can lead to poor results, as the example in  \figurename~\ref{fig:scvertex}(a) shows.
In some cases, one can do better by considering the \emph{maximal extensions} of such segments: the maximal extension of a segment $s$ is the longest segment containing $s$ that has both endpoints on edges of \N.
Yang~\cite{yang} showed that for paths, it is enough to consider only maximal extensions, a fact that allowed him to obtain an additive approximation for this class of graphs.
Unfortunately, as \figurename~\ref{fig:scvertex}(b) illustrates, maximal extensions do not work anymore as soon as $\N$ is a tree: an extension of a segment can lead to a worse diameter than the segment itself.
However, in this section, we show that if one
considers \emph{all} extensions of segments defined by two vertices of $\N$, then it is possible to guarantee
an approximation factor for general networks.

\begin{figure}[t]
\begin{center}
\includegraphics{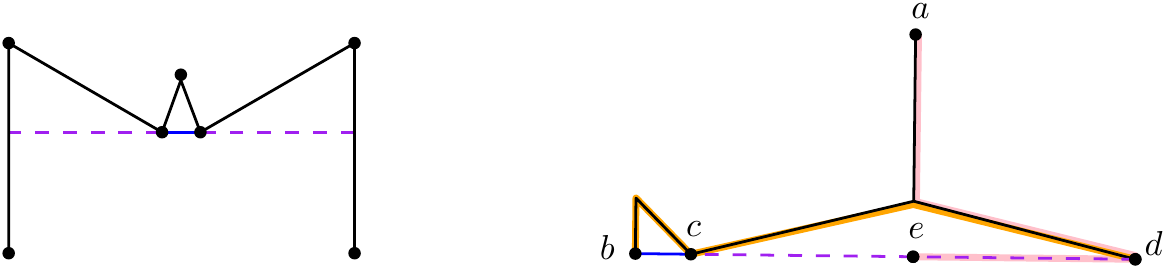}
\end{center}
\caption{(a) The optimal shortcut is the dashed purple segment. The blue segment (and any other segment between two vertices) gives a larger diameter.
(b) the original diameter is given by the orange path. The best shortcut connecting two vertices is $bc$. Contrary to intuition, extending $bc$ to $bd$ worsens the diameter, which becomes given by points $a$ and $e$ (pink path).}\label{fig:scvertex}
\end{figure}

Let $\mathcal{S}$ be the infinite set of segments with endpoints in $\N$, and let $\mathcal{S}_2 \subset \mathcal{S}$ be the subset of segments of $\mathcal{S}$ that contain two vertices of $\N$. The following proposition states that set $\mathcal{S}_2$ is an approximation of $\mathcal{S}$.

\begin{proposition} \label{prop:4d}
Let $\rho$ be the largest edge length in $\N$. Then, $$\min_{s \in \mathcal{S}}{\rm diam}(\N \cup s) \leq \min_{s \in \mathcal{S}_2}{\rm diam}(\N \cup s) \leq \min_{s \in \mathcal{S}}{{\rm diam}(\N \cup s)}+4\rho.$$
\end{proposition}

\begin{proof}
The first inequality is straightforward. For the second, it suffices to prove that given $s=pq \in \mathcal{S}$ there exists $s' \in \mathcal{S}_2$ such that ${\rm diam}(\N \cup s') \leq {\rm diam}(\N \cup s)+4\rho$.

Segment $s$ may cross several faces of $\N$, refer to \figurename~\ref{fig:aproximating}(a).

Consider the first and the last ones, say $\mathcal{F}_1$ and $\mathcal{F}_2$, together with the vertices of $\N$ that are adjacent to $p$ and $q$ in those faces: $u,v$ in $\mathcal{F}_1$ and $u',v'$ in $\mathcal{F}_2$. Let $V_1$ be the vertices of $\N$ in the quadrilateral $upqu'$ (including $u$ and $u'$), and let $C_1$ be its convex hull. Analogously, we have $V_2$ and $C_2$ for the quadrilateral $vpqv'$. Note that both convex hulls may have one point in common. Extend one of the common internal tangents of $C_1$ and $C_2$ giving rise to a segment $s'$ with endpoints on two of the edges of $\mathcal{F}_1$ and $\mathcal{F}_2$ containing points $p$ and $q$. Observe that $s'$ intersects all the edges of $\N$ that are crossed by $s$. This construction allows us to show that, for any two points $a,b\in \N$, the length of the shortest path between $a$ and $b$ that uses $s'$ is at most $4\rho$ plus the corresponding length but using $s$. To do this, we first use the triangle inequality to compare the lengths of the used portions of segments $s$ and $s'$, which gives a difference of $2\rho$, and then we add the two distances indicated in \figurename~\ref{fig:aproximating}(b). A similar argument is used for $a\in s'$ and $b\in \N$.\end{proof}

\begin{figure}
    \centering
    \includegraphics{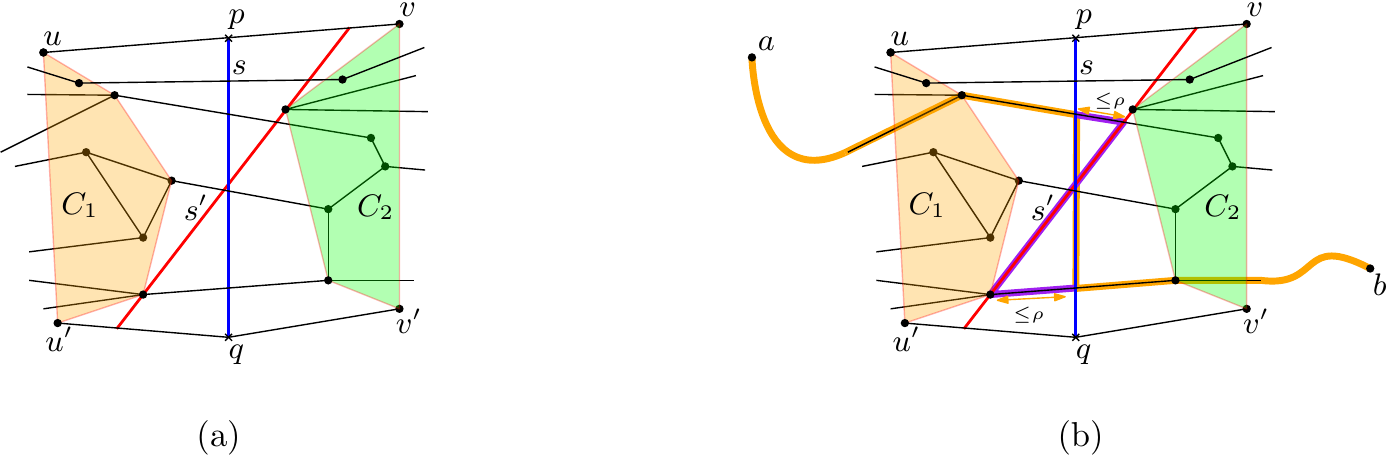}
    \caption{(a) Approximating a shortcut $s$ with a segment $s' \in \mathcal{S}_2$. (b) Using $s'$ instead $s$ to go from $a$ to $b$ causes a detour of at most $4\rho$ (purple path).}
    \label{fig:aproximating}
\end{figure}

The collection $\mathcal{S}_2$ is finite but quite large, it has size $O(n^4)$, which gives a time complexity of $O(n^6)$ to compute the optimal among the segments in $\mathcal{S}_2$. Indeed, there are $O(n^2)$ possible extensions per each pair of vertices, and for each of them one needs to compute the diameter from scratch in $O(n^2)$ time~\cite{hln-tccs-1991}. 

We would like to find a small subset of $\mathcal{S}_2$ that preserves the property in Proposition \ref{prop:4d}. Ideally, we would like to consider not all the extensions of a segment with endpoints in $V(\N)$ (that is exactly  $\mathcal{S}_2$), but only the best extension for each segment. 
Unfortunately, this appears rather difficult, as extensions of a segment do not seem to behave monotonically: already for a tree with a single vertex of degree larger than two, it may happen that an extension of a segment gives a worse diameter than the segment itself, see \figurename~\ref{fig:scvertex}(b).
However, we show next that we can speed-up the computation of the diameter for each extension in $\mathcal{S}_2$, saving a nearly-linear factor in the total running time.

Given a segment $s'=p'q'$, let $r$ be the ray starting at $p'$ and containing $s'$, and let $\mathcal{P}=p_0, p_1, \ldots , p_k$ be the sorted list of intersection points of $r$ with edges of $\N$. Segments $s_i=p'p_i$ are called extensions of $s'$ to the \emph{right}; the extensions to the \emph{left} are defined similarly.
Next we show how to speed-up the re-computation of the diameter of $\N \cup s_i$ as we insert $s_0,s_1, \dots s_k$, in that order.
To that end, we split the re-computation of distances into two parts:  distances from points on $s_i$ to points on $\N$,
and distances (in $\N \cup s_i)$ between two points on $\N$.

\begin{figure}[t]
\begin{center}
\begin{tabular}{ccccccc}
\includegraphics{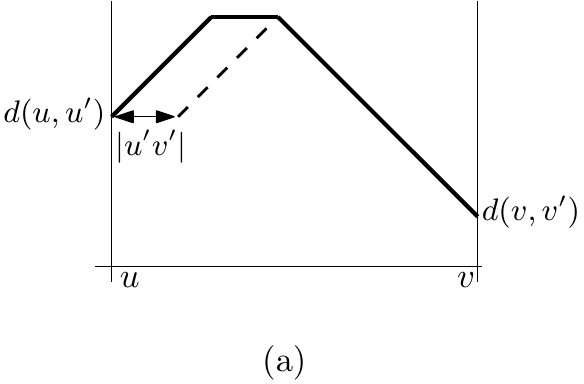}
& & & & & &
\includegraphics{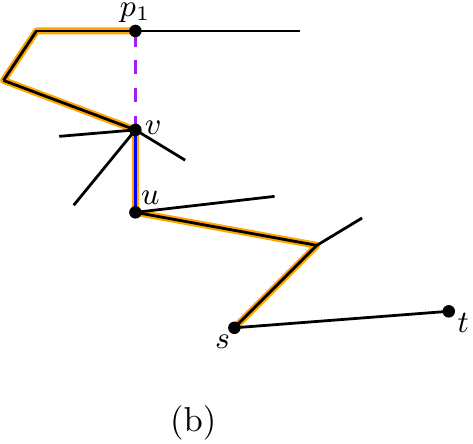}
\end{tabular}
\end{center}
\caption{(a) Function $\Phi_{uv}^{st}$. (b) If the distance from $p_1$ to $st$ decreases when adding $s_0=uv$, then, after adding $s_1=up_1$, it will become even smaller.}
\label{fig:phi_plus}
\end{figure}

\begin{lemma} \label{lem:es}
Let $u$ and $v$ be vertices of $\N$. It is possible to compute the eccentricities of all the extensions to the right  of segment $uv$ in $O(n^2)$ time.
\end{lemma}
\begin{proof}
As a preprocessing step, we store the distances from each vertex to all the other edges and the point at each edge attaining that maximum distance. This allows us to construct the functions $\Phi_{uv}^{st}:[0,1] \rightarrow \mathbb{R}^+$ that encode the information of the maximum distance from each point on an edge $uv$ to an edge $st$ (see~\cite{hln-tccs-1991}). Their
shape is as follows (see \figurename~\ref{fig:phi_plus}(a)).

Let $u'$ and $v'$ be the farthest points to, respectively, $u$ and $v$ in edge $st$.
Function $\Phi_{uv}^{st}$ increases uniformly from $0$ and from $1$ until the distance between both lines equals the distance between $u'$ and $v'$, at that moment it stabilizes horizontally. 

Therefore, knowing the farthest points $u'$ and $v'$ to $u$ and $v$ in the segment $st$ (and the distance between them), it is possible to build $\Phi_{uv}^{st}$ in constant time. The main idea of the proof is that it is also possible to update each map $\Phi_{uv}^{st}$ for each extension of a segment in constant time. Observe that $\Phi_{uv}^{st}$ encloses the information of the largest distance from any point of $uv$ to the segment $st$.

In a first step, we consider the addition of the segment $s_0=uv$. As $u$ and $v$ are on the network $\N$, they belong to some edges $g$ and $g'$, and we use the information of $\Phi_g^{st}$ and $\Phi_{g'}^{st}$ to find the largest distance from $u$ and $v$ to $st$ (in the network $\N$). With that information, we compute $\Phi_{uv}^{st}$ in constant time.
Thus, the maximum eccentricity of the edge $s_0=uv$ can be computed in linear time.

 Note that, by building the map $\Phi_{uv}^{st}$, it is possible to detect whether vertex $v$ changes its eccentricity with respect to $st$. Thus, we update the values of the distances from $vp_0$ to all the other edges, and the point on each edge giving that maximum distance (again, in linear time).

The addition of $s_0$ may change the eccentricity of $p_1$ with respect to some other edge (and the same with any of the other $p_i$'s), but we do not need to update that information at this moment.
Indeed, if the distance from $p_1$ to $st$ changes when adding $s_0$, it can only decrease.
Hence, the addition of $s_1$ will make that distance even smaller; see \figurename~\ref{fig:phi_plus}(b). Thus, in step $i$, we only need to update the information of the new vertex $p_i$, since by adding $s_i$, the value of $p_{i+1}$ will be updated.
\end{proof}

\begin{lemma} \label{lem:extension}
Let $u$ and $v$ be vertices of $\N$. It is possible to find the extension $s$ of segment $uv$ that minimizes ${\rm diam}(\N \cup s)$ in $O(n^3 \log n)$ time.
\end{lemma}
\begin{proof}
The value of ${\rm diam}(\N \cup s)$ can be computed by calculating the eccentricity of segment $s$ and comparing with the eccentricities in $\N \cup s$ of all the points in \N.
Thus, for each extension $s'$ of $uv$ to the left, we compute the eccentricities $E(i)$ of all its extensions $s_i$ to the right using Lemma \ref{lem:es}, in $O(n^2)$ time.
Let $N(i)$ be the maximum distance in $\N \cup s_i$ between pairs of points in \N.
Our goal is to compute $\min_i \max \{ E(i), N(i) \}$.

Since $N(i)$ is a decreasing function as $i$ grows, we do not need to compute $N(i)$ for all values of $i$, but only for those $i$ for which $E(i)$ is maximal: such that there is no $j>i$ with $E(j)<E(i)$ (see \figurename~\ref{fig:domina}).
Therefore, we can look for  that minimum by binary search, computing $N(i)$ only for $O(\log n)$ values of $i$.
Using \cite{f-fafs-1987}, we can update the distances between vertices in quadratic time and then compute $N(i)$ also in quadratic time (the distance between pairs edge-edge and vertex-edge can be computed in constant time knowing the distance between vertices), giving a total time of $O(n^3 \log n)$.
\end{proof}

\begin{figure}[tb]
\begin{center}
\includegraphics{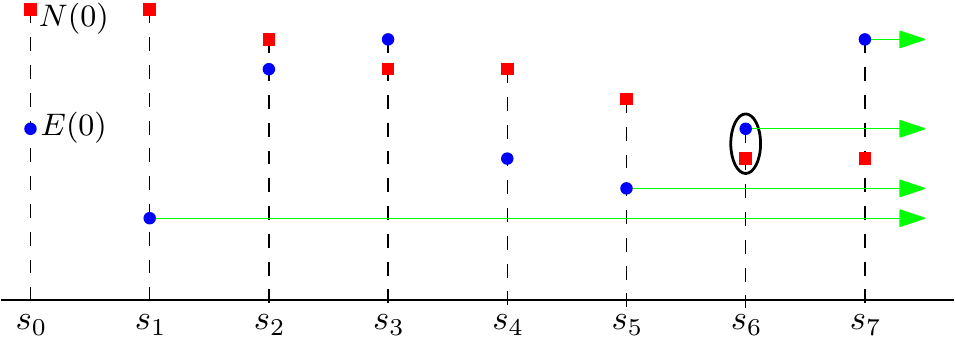}
\end{center}
\caption{For each segment extension, we consider two values:  $E(i)$---maximum eccentricity on $s_i$---(blue points), and $N(i)$---maximum eccentricity in $\N \cup s_i$ for points in \N---(red squares). For each $s_i$, the diameter of $\N \cup s_i$ is given by the maximum of these two values.
Only the blue points with green arrows must be tested (they are the maximal blue points). }\label{fig:domina}
\end{figure}

We thus obtain the main result in this section.

\begin{theorem} \label{th:approx}
Let $\rho$ be the largest edge length in a network $\N$. Then, it is possible to find a segment $s'$ such that $ {\rm diam}(\N \cup s') \leq  \min_{s \in \mathcal{S}}{\rm diam}(\N \cup s)+4\rho$ in $O(n^5 \log n)$ time.
\end{theorem}

This result immediately gives a simple approximation algorithm: subdivide each edge in \N by adding dummy vertices such that the largest resulting edge length is $\varepsilon$.
Then the previous theorem implies the following result, which is a generalization to general networks of the result for paths presented in~\cite[Theorem 8.1]{yang}.

\begin{corollary}
Let $\rho$ be the largest edge length in a network $\N$. Then, for any $0 < \varepsilon < \rho/2$ it is possible to find a segment $s'$ such that $  {\rm diam}(\N \cup s') \leq  \min_{s \in \mathcal{S}}{\rm diam}(\N \cup s)+4 \varepsilon$ in $O((n \rho/\varepsilon)^5 \log (n \rho))$ time.
\end{corollary}

\section{Path networks}\label{sec:path}

In the remaining, we focus on networks that are paths.
To illustrate the complexity of this seemingly simple setting, we begin by observing that the insertion of a shortcut into a path can create a quadratic number of diametral pairs;
as illustrated in the construction in \figurename~\ref{fig:quadratic_pairs}. It consists of $\Theta(n)$ spikes placed symmetrically with respect to the midpoint of the shortcut, denoted with $o$.
After inserting $pq$, each spike forms a face with a cycle of length roughly its height.
The spikes are spaced by one unit each, while their heights are set such that the distance from $o$ to the top of the spike is always the same, namely $|pq|/2$.
In this way, for any two spike tops $p_i$ and $q_j$ on the left and right of $o$, respectively, the distance between $p_i$ and $p_j$ on $\PP \cup pq$ is always equal to $|pq|$, which is also the diameter of $\PP \cup pq$.

\begin{figure}[t]
\begin{center}
\includegraphics{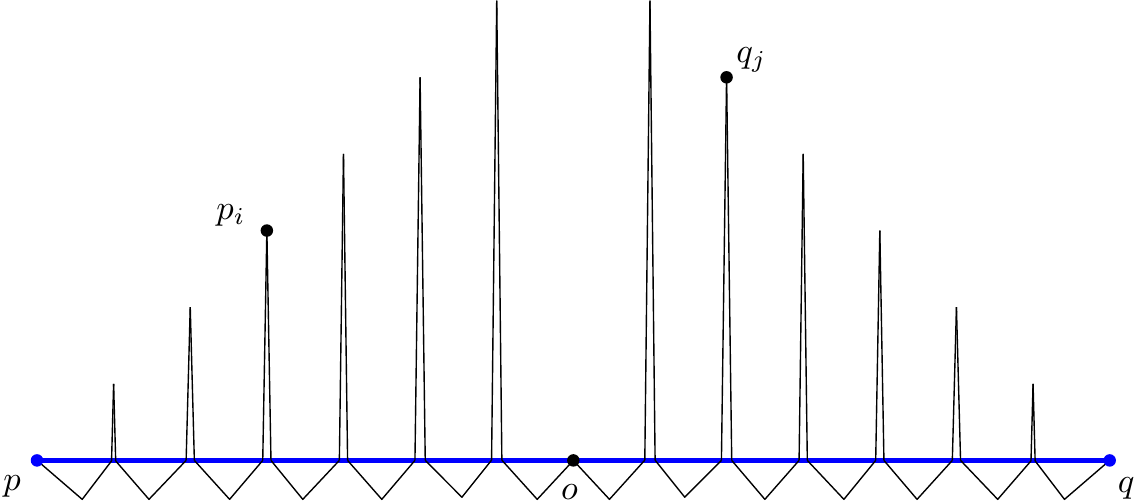}
\end{center}
\caption{Construction showing that the insertion of a shortcut $pq$ can create $\Theta(n^2)$ diametral pairs. The distance between the top of one spike on the left of $o$ and one on its right, like $p_i$ and $q_j$, is $|pq|$, and equals the diameter of $\PP \cup pq$.}
\label{fig:quadratic_pairs}
\end{figure}

\subsection{Diameter after inserting a shortcut}\label{sec:diameter}

The diameter of $\PP$ can be immediately computed in linear time, however, the addition of a shortcut $s$ can create a linear number of new faces, thus in principle it is not clear whether $\diam( \PP \cup s)$ can be computed in linear time, i.e., without computing the diameter between each pair of faces. The main result in this section is that this is still possible.

Path networks have the nice property that the maximal extension of an optimal shortcut is also optimal~\cite{yang}.
Thus, we can assume that $s=pq$ is maximal and horizontal.
The insertion of $s$ splits \PP into polygonal chains, which bound the different faces created.
Our goal is to compute the pair of chains that have maximum distance in  $\PP \cup s$.

We number the chains from $0$ to $m$ in the order of their left endpoints from left to right along $s$ (using right endpoints do disambiguate).
Except for possibly the first and last, all chains have both endpoints on $s$.
For the $i$th chain $C_i$, we denote its left and right endpoints by $p_i^l$ and $p_i^r$, respectively.
If the first vertex of $\PP$ is not on $s$, we consider the path from its first vertex to the first intersection of $\PP$ with $s$ as a degenerate loop chain with equal left and right endpoints on $s$ (analogous for the last vertex of $\PP$).

Refer to \figurename~\ref{fig:pathdiam-basics}.

Let $|C_i|$ be the length of $C_i$, let $L_i=|pp^{l}_i|$ and $R_i=|p^{r}_iq|$, and
let $s_i$ denote the segment $p_i^lp_i^r$.
Note that $C_i \cup s_i$ forms a cycle.
We use $D_i$ for the distance on $\PP \cup s$ from $p_i^l$ to its furthest point $\bar{p}_i^l$ on $C_i \cup s_i$ (i.e., $D_i$ is the semiperimeter of $C_i \cup s_i$); see \figurename~\ref{fig:pathdiam-basics}(b).

\begin{figure}[t]
\begin{center}
\includegraphics{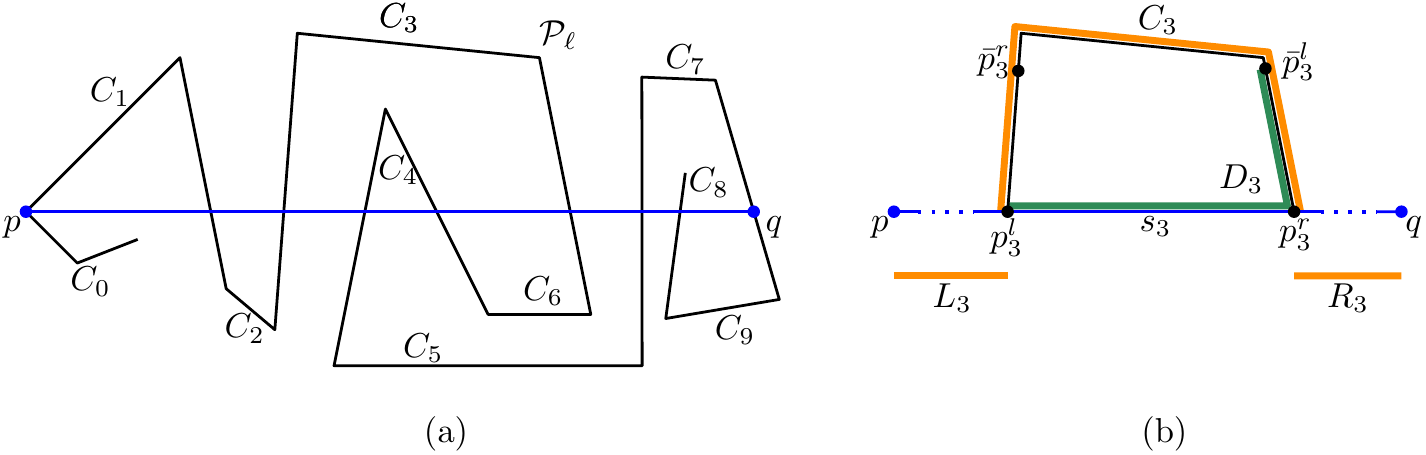}
\end{center}
\caption{(a) Chains created by $s$; $C_0$ and $C_8$ are degenerate chains.
(b) Detail for chain $C_3$, showing the cycle formed by $C_3 \cup s_3$. Thick lines are used here to denote distances.}
\label{fig:pathdiam-basics}
\end{figure}

We make some basic observations about the diameter between two chains, depending on their relative position, as shown in Figure \ref{fig:pathdiam-cases}.
They reveal the key property of the problem: the linear ordering between chains induced by $s$ defines uniquely how the diameter between two chains is achieved.

\begin{figure}[t]
\centering
\includegraphics{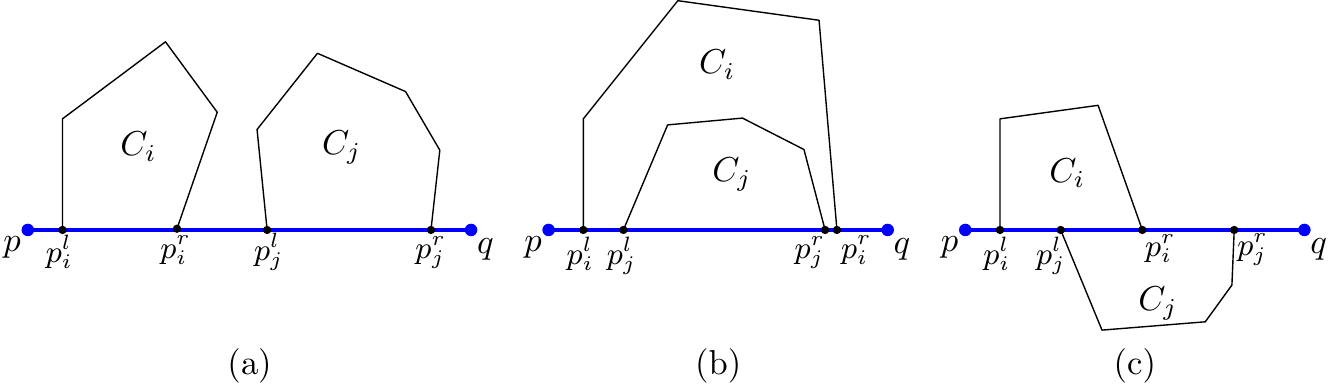}
\caption{Three cases when computing the diameter between two chains: (a) disjoint, (b) nested, and (c) overlapping.}
\label{fig:pathdiam-cases}
\end{figure}

\begin{observation}[Disjoint chains]
\label{obs:disjoint_chains}
Let $C_i,C_j$ be two chains of $(\PP \cup s)$ with $s_i \cap s_j = \emptyset$ and $s_i$ to the left of $s_j$.
The diameter of $C_i \cup p_i^l p_j^r \cup C_j$ is $D_i+|p_i^r p_j^l|+D_j = D_i+R_i-R_j-|s_j|+D_j$.
\end{observation}

\begin{proof}
Chain $C_i$ is completely to the left of $C_j$ (see Figure \ref{fig:pathdiam-cases}(a)), thus any path from a point in $C_i$ to one in $C_j$ must include segment $p_i^r p_j^l$.
Since the furthest point in $C_i$ from $p^r_i$ is $\bar{p}^r_i$, and the same for $p^l_j$ and $\bar{p}^l_j$, the result follows by definition of $D_i$ and $D_j$, and by observing that $|p_i^r p_j^l| = R_i-R_j-|s_j|$.
\end{proof}

\begin{observation}[Nested chains]
\label{obs:nested_chains}
Let $C_i,C_j$ be two chains of $(\PP \cup s)$ with $s_j \subset s_i$.
The diameter of $C_i \cup s_i \cup C_j$ is $\frac12(|C_i|+|p^l_i p^l_j|+|p^r_i p^r_j|+|C_j|) = \frac12(|C_i|+L_j-L_i+R_j-R_i+|C_j|)$.
\end{observation}

\begin{proof}
In this case the endpoints of the chain $C_j$ are nested inside those of $C_i$ (see Figure \ref{fig:pathdiam-cases}(b)), hence the two chains behave effectively like a cycle, because no diametral pair can have a point on segment $s_j$, since $|C_j|>|s_j|$.
The cycle is formed by chains $C_i$, $C_j$, and the two subsegments of $s_i \setminus s_j$.
The diameter of any cycle is half of its perimeter, which in this case amounts to $\frac12(|C_i|+|s_i|-|s_j|+|C_j|)$.
The second equality follows from the fact that $|p^l_i p^l_j|+|p^r_i p^r_j| = L_j-L_i+R_j-R_i$.
\end{proof}

\begin{observation}[Overlapping chains]
\label{obs:overlapping_chains}
Let $C_i,C_j$ be two chains of $(\PP \cup s)$ with $s_i \cap s_j \neq \emptyset$, $p_i^l \notin s_j$ and $p_j^r \notin s_i$.
The diameter of $C_i \cup p_i^l p_j^r \cup C_j$ is
 $\frac12(|C_i|+|p^l_i p^l_j|+|p^r_i p^r_j|+|C_j|)  = \frac12(|C_i|+L_j-L_i+R_i-R_j+|C_j|)$.
\end{observation}

\begin{proof}
As in the previous case, the two chains behave effectively like a cycle (see Figure \ref{fig:pathdiam-cases}(c)), in this case formed by $C_i$, $C_j$ and the segments $p^l_i p^l_j$ and $p^r_i p^r_j$.
The second equality follows from the fact that $|p^l_i p^l_j| = L_j-L_i$ and $|p^r_i p^r_j| = R_i-R_j$.
\end{proof}

Note that, while in the case of Figure \ref{fig:pathdiam-cases}(a) the diameter is achieved by a unique pair of points, that is not the case in situations (b) and (c), for which an infinite number of diametral pairs of points may exist (as in any cycle). Also, note that expressions for the diameter in Observations \ref{obs:nested_chains} and \ref{obs:overlapping_chains} are the same except for adding up either $R_j-R_i$ or $R_i-R_j$.
This difference only exists to differentiate between the two possible orders of the right endpoints of the two chains.

The algorithm for computing $\diam( \PP \cup s)$ in linear time starts by going along $\PP$ and computing all intersections with $s$ in the order of $\PP$. Then we apply a linear-time algorithm for Jordan sorting~\cite{FUNG199085} to obtain the intersections in the order along $s$, say, from left to right. Within the same running time we can compute the values $C_k$ and $s_k$.
Next, we sweep the endpoints of the chains along $s$ to compute, for each chain $C_k$, its furthest chain from the ones seen so far.
To that end, certain information is computed and stored:
\begin{enumerate}
\item The furthest chain from $C_k$ to the left, given by $\argmax_{0\leq i < k} \alpha_i$, where $\alpha_i=D_i+R_i$.
Similarly, we store the furthest chain to the right.
\item The furthest chain nested inside $C_k$. This is given by $\argmax_{j \in N_{k}} \beta_j$, where $\beta_j=|C_j|+L_j+R_j$ and $N_k$ is the set of indices of all chains nested inside $C_k$.
\item The furthest chain with one endpoint in $C_k$, and one outside: given by
$\argmax_{j \in O^r_{k}} \gamma_j$, where $\gamma_j=|C_j|+L_j-R_j$ and $O^r_k$ is the set of indices of all overlapping chains with their left endpoint inside $C_k$ and their right endpoint outside.
Similarly, we store those with their left endpoints outside and the right one inside of $s_k$.
\end{enumerate}

The computation of the information in (1) is straightforward when sweeping along $s$, say, from left to right.
We just maintain the largest value of $\alpha_i$ seen so far as we sweep.
The case of nested or overlapping chains, which is explained next, is more complicated because one needs the maximum restricted to those chains that are contained or overlap with $C_k$.

Suppose that $C_i$ starts to the left of $C_j$ (the other case is analogous). We use a data structure for range minimum queries~\cite{bfpss-lcatdag-05,bv-rstpds-93}.
This allows to preprocess an array $A$ in linear time in order to find the maximum value in any subarray $A[a,b]$ in $O(1)$ time.
In our context, we need two such data structures.
We use arrays $A_\text{n}$ and $A_\text{o}$ to store the maximum $\beta$ and $\gamma$ values defined above for the chains that are nested and overlapping, respectively.
Each array has one position for each endpoint of a chain, thus $2m$ in total.
The positions are as they appear sorted along $s$, from left to right.
Refer to \figurename~\ref{fig:pathdiam-arrays}.

For a chain $C_j$, the position corresponding to its left endpoint has a value equal to $\beta_j$ in the array $A_\text{n}$, and value $\gamma_j$ in array $A_\text{o}$.
The values corresponding to the right endpoints of the chains are not used, i.e., they have value $-\infty$, in both arrays.
At each array position, we also store pointers to the corresponding chains.

\begin{figure}[t]
\centering
\includegraphics{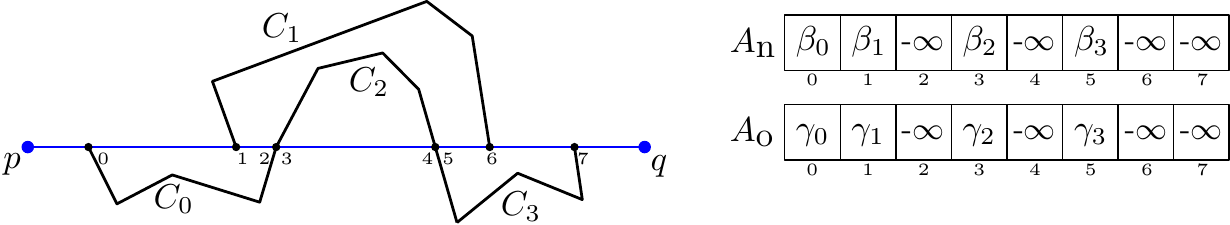}
\caption{Example of arrays with distances to chains that are nested ($A_\text{n}$) and overlapping ($A_o$).}
\label{fig:pathdiam-arrays}
\end{figure}

To find the nested or overlapping chain furthest from $C_i$ we would like to perform one maximum range query in $A_\text{n}$ and one in $A_\text{o}$, in both cases with a subarray corresponding to the interval between the endpoints of $C_i$.
The goal is to use these queries to obtain the furthest chain of each type: nested and overlapping.
However, there is an issue.
In the way $A_\text{n}$ and $A_\text{o}$ are defined, the result of a range query cannot distinguish between nested or overlapping chains, it necessarily searches in both sets (i.e., $N_i \cup O^r_i$).
Fortunately, the geometry of the problem guarantees that we can still use the result obtained, as we show next.
The following lemma shows that if the furthest face is associated to a $\beta_i$ value, then it must be nested, and similarly, if it is associated to a $\gamma_i$ value, it must be overlapping.

\begin{lemma}
\label{lem:unionfaces}
Let $C_k$ be a chain with distance to $C_i$ equal to
$$d^* =\max \{ \max_{j \in (N_i \cup O^r_i)} |C_i|-L_i-R_i+\beta_j , \max_{j \in (N_i \cup O^r_i)} |C_i|-L_i+R_i+ \gamma_j \}.$$ Then it holds that:
\begin{enumerate}
\item[(i)] if $d^* = |C_i|-L_i-R_i+\beta_k$, then $k\in N_i$;
\item[(ii)] if $d^* =  |C_i|-L_i+R_i+ \gamma_k$, then $k\in O^r_i$.
\end{enumerate}
\end{lemma}

\begin{proof}
We prove (i); an analogous argument proves (ii). Suppose that $d^* = |C_i|-L_i-R_i+\beta_k$, but $k\in O^r_i$, i.e., $C_k$ is overlapping.
Further, to the right of $C_i$ we have $R_i > R_k$ and thus $R_k-R_i < 0 < R_i-R_k$. Hence, using this and the definitions of $d^*$ and $\beta_k$, we have
$$d^*= |C_i|-L_i-R_i+(|C_k|+L_k+R_k) < |C_i|-L_i+R_i+|C_k|+L_k-R_k = |C_i|-L_i+R_i + \gamma_k.$$
In other words, the value associated to $\gamma_k$ gives a larger distance than the one associated to $\beta_k$, contradicting the optimality of $d^*$.
\end{proof}

Therefore, when processing a chain $C_k$, we perform one maximum range query in $A_\text{n}$ and one in $A_\text{o}$, and keep the maximum of those two values.
Lemma~\ref{lem:unionfaces} guarantees that the associated chain is the furthest one that is either nested or overlapping.
Proceeding in an analogous way for the chains that are overlapping with one endpoint to the left of $C_k$, the furthest face from $C_k$ of any of the three types (disjoint, nested, overlapping) can be found in $O(1)$ time, and the maximum distance between two chains can thus be found in linear time.

\begin{theorem}
For every path $\PP$ with $n$ vertices and a shortcut $s$, it is possible to compute the diameter of $( \PP \cup s)$ in $\Theta(n)$ time.
\end{theorem}

It is worth noting that the ideas used in this section do not extend to networks that are trees.
As soon as $\N$ is a tree, the insertion of $s$ creates chains that can have several ways to connect in $\N \cup s$, making it impossible to know, a priori, the expression of their distance as we did for paths in Observations~\ref{obs:disjoint_chains}--\ref{obs:overlapping_chains}.

\subsection{Optimal horizontal shortcuts}

The observations in Section \ref{sec:diameter} also give us a way to compute an optimal horizontal shortcut for a path considerably faster than using the general method in Section \ref{general_networks}. After a suitable rotation, this allows to find an optimal shortcut of any fixed orientation.

Assume as in Section \ref{sec:diameter} that shortcuts are horizontal and maximal, so they can be treated as horizontal lines. Now,
consider the vertices in $\PP$ sorted increasingly by $y$-coordinate, and let $y_a,y_b$, with $y_a < y_b$, be the $y$-coordinates of two consecutive vertices in that order. Observations \ref{obs:disjoint_chains}--\ref{obs:overlapping_chains} are stated in terms of chains, but they also apply to faces.
Indeed, they imply that the distance between any two faces $f_i$ and $f_j$ is a linear function $d_{ij}(y)$ for $y_a \leq y \leq y_b$. Thus, each face is associated with $k-1$ lines in 2D where $k$ is the total number of faces, leading to a set ${\cal L}$ of $\Theta(k^2)$ lines (note that $k=O(n)$). 

The optimal shortcut over all $y \in [y_a,y_b]$ is given by the minimum of the upper envelope of ${\cal L}$, which can be computed in $O(k^2 \log k)$ time~\cite{BergCKO08}. If this is done with each of the $n-1$ horizontal strips formed by consecutive vertices of $\N$, the optimal horizontal shortcut is obtained in total $O(n^3 \log n)$ time. 

 The preceding method can be improved if, instead of computing from scratch the upper envelope of ${\cal L}$ at each horizontal strip, we maintain the upper envelope between consecutive strips and only add or remove the lines that change when going from one strip to the next one. The changes between two consecutive strips are of three types:
\begin{enumerate}
\item[(i)] one of the two line segments bounding a face within the strip changes;
\item[(ii)] a face ends;
\item[(iii)] a new face appears.
\end{enumerate}
In the worst case, $n-1$ lines are removed from ${\cal L}$ and another $n-1$ lines are added to ${\cal L}$.
Maintaining the upper envelope of $N$ lines is equivalent to maintaining the convex hull of $N$ points in 2D, which can be done in amortized  $O(\log N)$ time per insert/delete operation with a data structure of size $O(N)$~\cite{bj-dpch-02}. Since $N = O(n^2)$, we obtain the following theorem.

\begin{theorem}
For every path $\PP$ with $n$ vertices, it is possible to find an optimal horizontal shortcut in $O(n^2 \log n)$ time, using $O(n^2)$ space.
\end{theorem}

\subsection{Optimal simple shortcuts}

In this section we consider optimal \emph{simple} shortcuts, i.e., we restrict the possible shortcuts to those whose interior does not intersect \N.
We show that an optimal simple shortcut can be computed much faster if it exists. Note that one must distinguish between an \emph{optimal simple shortcut} and a \emph{simple optimal shortcut}. The first is a shortcut that is optimal in the set of simple shortcuts; this is different of being optimal in the set of all shortcuts and, in addition, being simple.

Interestingly, it is known that optimal simple shortcuts may not exist, even for paths~\cite{YangBell98} (e.g., when the only optimal shortcut goes through a vertex, see Figure \ref{fig:optimal_simple}(a)). It is not clear, however, what the conditions for a network $\N$ to have an optimal shortcut are, even restricted to simple shortcuts. The following proposition is a first approach to this question.

\begin{figure}[ht]
 \begin{center}
\includegraphics{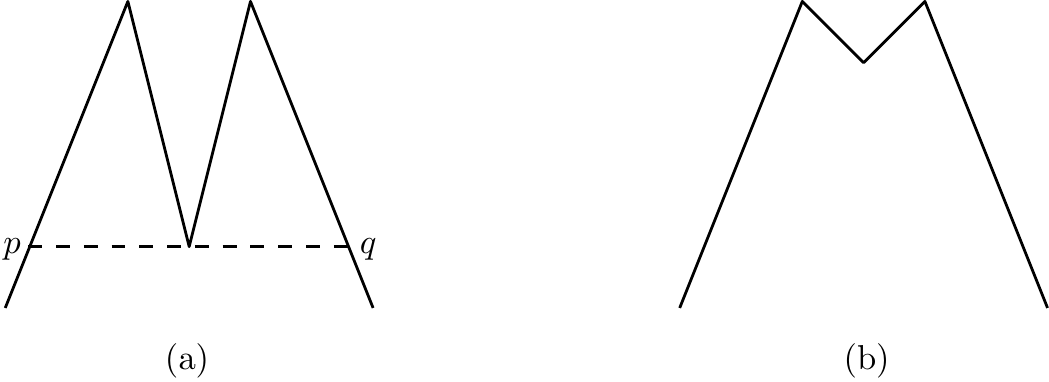}
\end{center}
\caption{(a) A network with no optimal simple shortcut: segment $pq$ can be approached as much as desired with simple shortcuts. (b) $\N$ admits an optimal simple shortcut and $\overline{\mathcal{N}}$ has non-convex faces.}\label{fig:optimal_simple}
\end{figure}

\begin{proposition}\label{prop:optimal_simple}
Let $\mathcal{N}$ be a network whose locus $\N$ admits a simple shortcut, and let $\overline{\mathcal{N}}$ be the network resulting from adding to $\mathcal{N}$ all edges of the convex hull of $V(\mathcal{N})$. If all faces of $\overline{\mathcal{N}}$ are convex, then $\N$ has an optimal simple shortcut.
\end{proposition}

\begin{proof}
Suppose on the contrary that there is no optimal simple shortcut for $\N$. Then, there is a sequence $\{s_n\}$ of simple shortcuts such that $d_{n+1}<d_{n}$ where $d_n={\rm diam}(\N \cup s_{n})$. Let $s_n=p_nq_n$.
We may assume, without loss of generality, that $\{s_n\}$ is contained in a face $\mathcal{F}$ of $\overline{\mathcal{N}}$ as there is a finite number of faces, at least one of which contains a subsequence of $\{s_n\}$. The same argument on the number of edges of a face lets us assume that for every $n$, points $p_n$ and $q_n$ are, respectively, on edge $e$ and edge $e'$ of $\N$.

Since $e$ is a compact set, sequence $\{p_n\}$ has a convergent subsequence $\{p_{n'}\}$. The corresponding sequence $\{q_{n'}\}$ is not necessarily convergent, but edge $e'$ is also a compact set which implies that $\{q_{n'}\}$ contains a convergent subsequence $\{q_m\}$. Let $\{s_m\}=\{p_mq_m\}$, $\lim p_m=p$, $\lim q_m=q$, and $s=pq$; we will write, with some abuse of notation, $\lim s_m=s$ . Next, we prove that $d_s={\rm diam}(\N \cup s)<d_n$ for every $n$.

Fixed $n$, consider two segments $s_{m_0},s_{m_1} \subset \{s_m\}$ with $n<m_0<m_1$, and take $\varepsilon>0$ such that $d_{m_0}-d_{m_1}>\varepsilon$; such segments do exist as $\lim s_m=s$. Let $\mathcal{R}$ be the set of points in face $\mathcal{F}$ whose distance to $p$ or $q$ is smaller than $\varepsilon/8$. Any segment $s_{m'}$ with endpoints $p'\in e\cap \mathcal{R}$ and $q'\in e'\cap\mathcal{R}$ verifies that $|d_s-d_{m'}|<\varepsilon/2$. Indeed, $$|d_s-d_{m'}|<|p-p'|+|q-q'|+||pq|-|p'q'||<\varepsilon/8+\varepsilon/8+2\varepsilon/8=\varepsilon/2.$$
Since $\lim s_m=s$, there exists $m_2$ such that $s_m \in \mathcal{R}$ for all $m\geq m_2$. Let $m_3={\rm max}\{m_1,m_2\}$.
 If $d_s<d_{m_3}$ then $d_s<d_{m_3}+\varepsilon <d_{m_0}<d_n$. Otherwise, $$d_s<d_{m_3}+\varepsilon/2<d_{m_0}-\varepsilon/2<d_{m_0}<d_n.$$

Thus, for every sequence $\{s_n\}$ of simple shortcuts such that $d_{n+1}<d_{n}$ there is a shortcut $s$ satisfying that $d_s<d_{n}$ for all $n$. Further, $s$ is simple (all faces of $\overline{\mathcal{N}}$ are convex) and there is no optimal simple shortcut. This implies that $s$ must be an edge of ${\mathcal{N}}$ which contradicts the fact that, by definition, edges are not shortcuts. Note that there can be sequences for which the corresponding segment $s$ is an edge of $CH(V(\mathcal{N})\setminus E(\mathcal{N})$ (here $CH(V(\mathcal{N})$ denotes the convex hull of $V(\mathcal{N})$), but there is a finite number of such edges and so one would obtain an optimal simple shortcut among all those segments $s$.
\end{proof}

Figure \ref{fig:optimal_simple}(b) shows that the converse of Proposition \ref{prop:optimal_simple} is not true.

We now turn our attention to the computation of an optimal simple shortcut, if it exists.

Let $s=pq$ be a simple shortcut for a path $\PP$ with endpoints $u,v$. Suppose that point $p$ is closer to $u$ than $q$ along $\PP$; let $x=d(u,p)$ and $y=d(v,q)$. There is only one bounded face in $\PP\cup s$ whose boundary is a cycle $C(p,q)$. Let $\overline{p}$ and $\overline{q}$ be the farthest points from, respectively, $p$ and $q$ on $C(p,q)$, and let $z=(d_{\PP}(p,q)-|pq|)/2$. Note that $d(\overline{p},\overline{q})=|pq|$ and $z=d(p,\overline{q})=d(\overline{p},q)$.
See \figurename~\ref{fig-path1b}.

\begin{figure}[t]
\begin{center}
\includegraphics{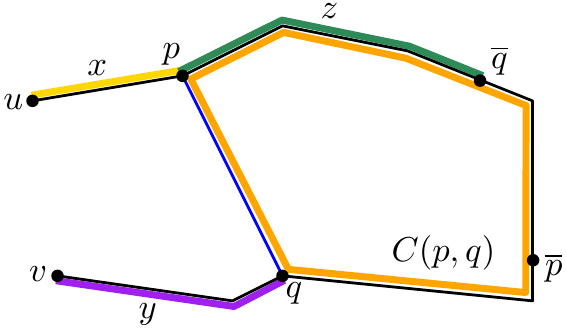}
\end{center}
\caption{Inserting a simple shortcut $pq$.}\label{fig-path1b}
\end{figure}

There are three candidates for diametral path in $\PP\cup s$ (see \cite{DBLP:conf/swat/CarufelMS16}):
\begin{enumerate}
\item The path from $u$ to $v$ via $s$ is diametral if and only if $z={\rm min}\{x,y,z\}$,
 \item the path from $u$ to $\overline{p}$ via $s$ is diametral if and only if $y={\rm min}\{x,y,z\}$,
\item the path from $v$ to $\overline{q}$ via $s$ is diametral if and only if $x={\rm min}\{x,y,z\}$.
\end{enumerate}

Thus, ${\rm diam}(\PP\cup s)\in \{x+y+|pq|, x+z+|pq|, y+z+|pq|\}$.

For the highway model, it was proved in~\cite{DBLP:conf/swat/CarufelMS16} that $\PP$ has an optimal shortcut satisfying $x=y$, which allows to compute it in linear time.
In the planar model the situation is more complicated but, in a similar fashion, we can prove the following lemma, which will lead to Theorem~\ref{th:oss}.

\begin{lemma} \label{lem:oss}
Let $pq$ be an optimal simple shortcut for $\PP$. The following statements hold.
\begin{enumerate}
\item If neither $p$ nor $q$ are vertices of $\PP$ then $x=y=z$.
\item If $p$ or $q$ are vertices of $\PP$ then the two smallest values among $x,y,z$ are equal.
\end{enumerate}
\end{lemma}

\begin{proof}
First note that increasing (resp., decreasing) the value of either $x$ or $y$ leads to a decrease (increase) in the value of $z$.

Suppose that neither $p$ nor $q$ are vertices of $\PP$ and assume on the contrary that $x$ is smaller (respectively, greater) than $y$ and $z$. Then point $p$ can be slightly moved away from $u$ (respectively, $q$ away from $v$) in order to obtain a shortcut $s'$ such that ${\rm diam}(\PP\cup s')<{\rm diam}(\PP\cup pq)$. Thus, shortcut $pq$ would not be optimal, which is a contradiction. The same argument applies for $y$. Finally, if the smallest value is $z$ then we can slightly move $p$ closer to $u$ and $q$ closer to $v$ while decreasing the diameter of the augmented network, again a contradiction.

Suppose now that $p$ is a vertex of $\PP$, i.e., $q$ is the only possible point to be moved in order to argue as above and reach a contradiction. If $x<y<z$ we can decrease $y$ by moving point $q$ closer to $v$ obtaining a ``better'' shortcut, a contradiction. If $x<z<y$ it suffices to increase $y$ in order to decrease $z$ to reach a contradiction. The remaining cases are analogous.
\end{proof}

\begin{theorem} \label{th:oss}
 It is possible to decide whether a path $\PP$ with $n$ vertices has an optimal simple shortcut and compute one (in case of existence) in $O(n^2)$ time.
\end{theorem}

\begin{proof}
We first determine all candidate segments $pq$ for an  optimal simple shortcut and also other relevant segments to decide the existence of such a shortcut. Two cases are distinguished depending on whether $p$ and $q$ are vertices.

\emph{Case 1.} $p,q\notin V(\PP)$: we add $O(n)$ extra vertices on each half of $\PP$, such that the $i$-th edges from each endpoint have the same length. Let $\mathcal{P}_{\ell}^1$ and $\mathcal{P}_{\ell}^2$ be the two resulting sub-paths where every edge $e_1\in \mathcal{P}_{\ell}^1$ has an associated edge $e_2\in\mathcal{P}_{\ell}^2$ of the same length, and vice-versa.

This partition of $\PP$ allows us to compute the points $p\in e_1 \in \mathcal{P}_{\ell}^1$ and $q \in e_2 \in\mathcal{P}_{\ell}^2$ where $x=y$ (and $|e_1|=|e_2|$).
Among these points and according to Lemma \ref{lem:oss}, we want those satisfying that $z=x$ $(=y)$.
Now, assume we have $p \in e_1$ and $q \in e_2$ such that $x=y=z$. Consider advancing a distance $\lambda$ along $e_1$ and $e_2$, from $p$ and $q$ to $p'$ and $q'$.
Since $d_{{P}_{\ell}}(p,q)=|\PP|-2x$ and $z=(d_{{P}_{\ell}}(p,q)-|pq|)/2$, we only need to compute the variation of $|pq|$ with respect to the variation $\lambda$ of $x$ (and $y$). \figurename~\ref{fig-path25}(a) shows the construction where $\alpha$ denotes the angle formed by the prolongations of $e_1$ and $e_2$, and $a$ and $b$ are, respectively, the distances from $p$ and $q$ to the intersection point of those prolongations. By the law of cosines we have
$$|p'q'|=\sqrt{(a\mp \lambda)^2+(b\mp \lambda)^2 - 2(a\mp\lambda)(b\mp\lambda)\cos \alpha}.$$
Note that the variation of $x$ might be in the direction of the intersection point $o$ of the prolongations of $e_1$ and $e_2$ or in the opposite direction and, in general, it does not lead to a parallel segment to $pq$ (see \figurename~\ref{fig-path25}(b)); in fact, segments $pq$ and $p'q'$ are parallel only if the triangle formed by $p, q$, and the intersection point $o$ is isosceles.

Therefore, by solving $O(n)$ quadratic equations, we obtain $O(n)$ candidate segments $pq$ with $p,q\notin V(\PP)$.

\begin{figure}[t]
\begin{center}
\begin{tabular}{ccccc}
\includegraphics{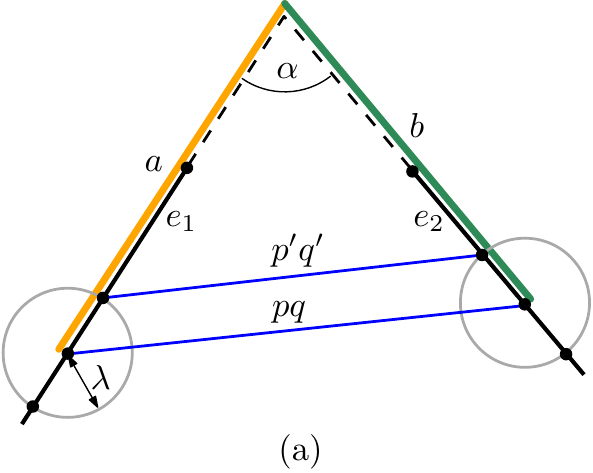}
& & & &
\includegraphics{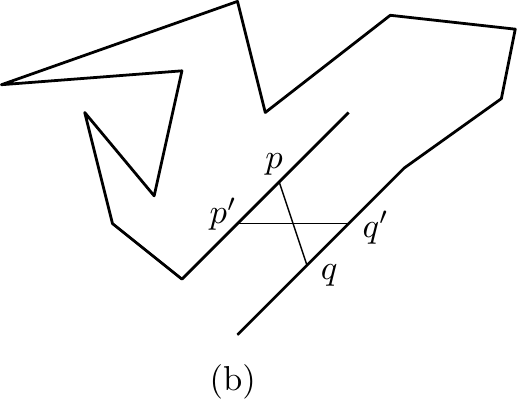}
\end{tabular}
\end{center}
\caption{(a) Updating $z$ after increasing $x$ and $y$ by $\lambda$. (b) Segments $pq$ and $p'q'$ are non-parallel.}\label{fig-path25}
\end{figure}

\emph{Case 2.}  $p$ or $q$ are vertices of $\PP$: according to Lemma \ref{lem:oss} our candidate segments must satisfy that the two smallest values among $x,y,z$ are equal. Clearly, there are $O(n^2)$ segments to consider with both endpoints being vertices; it remains to obtain those candidate segments with exactly one endpoint in $V(\PP)$.

For each vertex of $\mathcal{P}_{\ell}^1$ (i.e., $x$ is a fixed value), we have computed in case (1) the $O(n)$ candidate segments verifying that $x=y$ and so we must determine those for which $x=z$
(analogous for each vertex in $\mathcal{P}_{\ell}^2$ and $y=z$). This can be done as in case (1) where $\lambda$ is now the variation of $y$ since $x$ is fixed; the formula is $$|pq'|=\sqrt{a^2+(b\mp \lambda)^2 - 2a(b\mp\lambda)\cos \alpha}.$$ Note that $p$ is a fixed vertex of $\mathcal{P}_{\ell}^1$ and the corresponding $q$ and $q'$ are points on an edge of $\mathcal{P}_{\ell}^2$. Thus, there is a total of $O(n^2)$ candidate segments.

Besides the segments obtained in cases (1) and (2), we must consider other segments that are pivoting on a vertex $w$ of $\PP$ (see Figure \ref{fig-path4}) and such that the two smallest values among $x,y,z$ are equal. Note that in the preceding cases we have not computed those whose endpoints are not vertices and $z=x\neq y$ or $z=y\neq x$. This can be done applying twice the argument of case (2): to the left and to the right of the line shown in Figure \ref{fig-path4} passing through vertex $w$. The desired value $z$ is the sum of the two obtained values.

\begin{figure}[ht]
 \begin{center}
\includegraphics[width=0.4\textwidth]{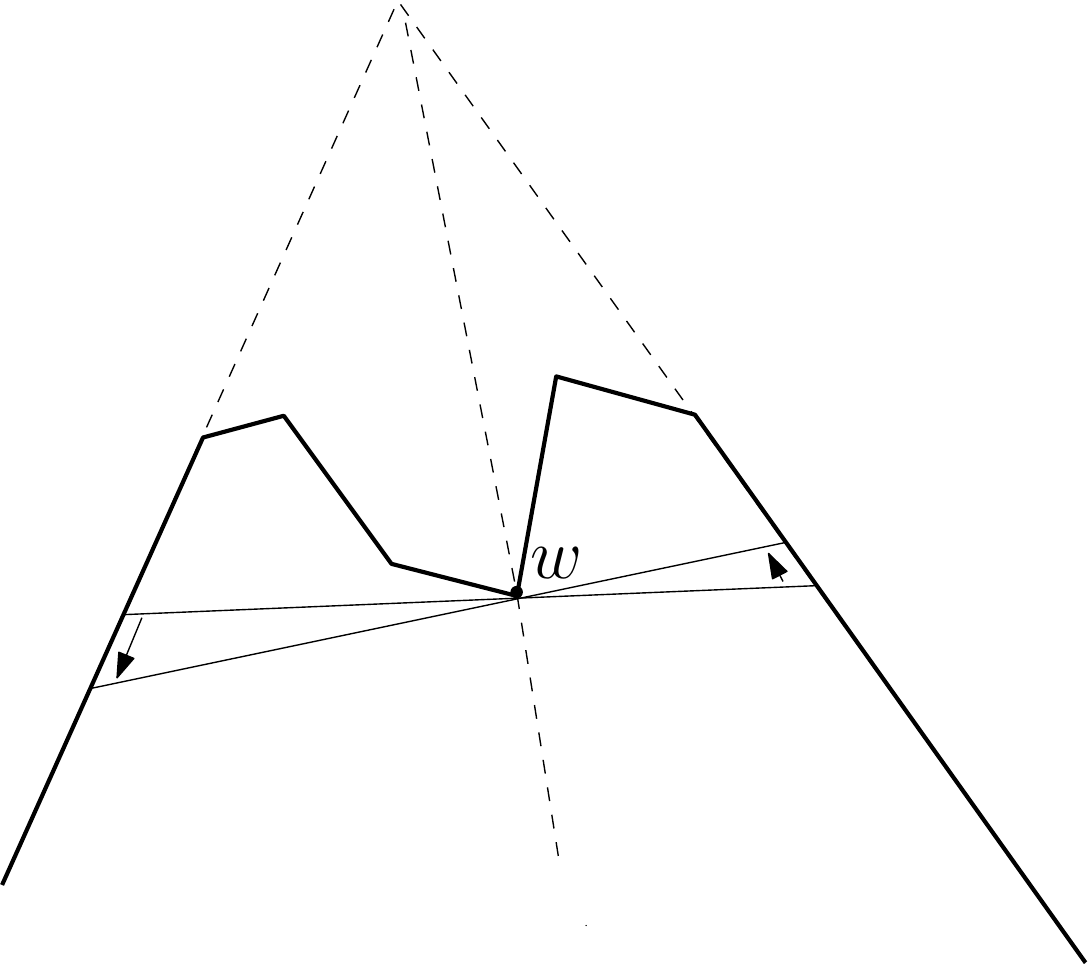}
\end{center}
\caption{Shortcut that is pivoting on vertex $w$.}\label{fig-path4}
\end{figure}

Finally, we classify all our segments, $O(n^2)$ in total,
into three sets: $\mathcal{S}$ of simple shortcuts, $\mathcal{L}$ of limit cases (the segment intersects $\PP$ on three points), and shortcuts that intersect $\PP$ on more than three points. There exists an optimal simple shortcut if and only if the minimum value of ${\rm diam}(\PP\cup pq)$ over $pq\in \mathcal{S}\cup\mathcal{L}$ is attained by a segment in $\mathcal{S}$.
\end{proof}

\section{Conclusions}\label{sec:conclusions}

In this paper we have presented the first results on the computation of optimal shortcuts for general networks in the planar model.
This can be seen as a particular variant of the road network design problem, where the problem is abstracted to its most fundamental geometric version.
Clearly, even in this restricted setting, the problem continues to be difficult and challenging.
We have shown that an optimal shortcut can be computed in polynomial time, and given a discretization of the problem that results in an approximation of the original continuous version.
Even though the discretization obtained is too large to be of practical use, it is interesting from a theoretical point of view, and hopefully will be useful to obtain smaller discretizations in the future.

We have also presented new results for paths, including how to quickly compute the diameter after inserting a shortcut, the computation of an optimal shortcut of fixed orientation, and of an optimal simple shortcut.
These are important first steps on a relevant and difficult problem, which leave many intriguing questions open.
The existence of a small discrete set of segments to approximate an optimal shortcut, or a fast algorithm to find an optimal shortcut for paths (any orientation), are some examples.

Finally, the questions studied in this paper but for optimal sets of $k>1$ shortcuts pose challenging open problems.

\paragraph{Acknowledgments.}

D.G. and R.S. were supported by project MTM2015-63791-R.
R.~S.\ was also supported by Gen.\ Cat.\ 2017SGR1640 and MINECO through the Ram{\'o}n y Cajal program.
A.M. was supported by project BFU2016-74975-P.

\bibliography{refs}

\end{document}